\documentclass{article}
\usepackage{amsmath,amsthm,amssymb,amsfonts}
\usepackage{verbatim}
\usepackage{graphicx}
\usepackage{stmaryrd} 
\usepackage{hyperref}

\usepackage[colorinlistoftodos]{todonotes}

\theoremstyle{plain}
\newtheorem{thm}{Theorem}
\newtheorem{lem}[thm]{Lemma}
\newtheorem{prop}[thm]{Proposition}

\theoremstyle{definition}
\newtheorem{definition}[thm]{Definition}
\newtheorem{exl}[thm]{Example}

\numberwithin{thm}{section}

\newcommand{\adj}{\leftrightarrow}
\newcommand{\adjeq}{\leftrightarroweq}

\def\Z{{\mathbb Z}}
\def\N{{\mathbb N}}
\def\R{{\mathbb R}}

\begin{document}
\title{Beyond the Hausdorff Metric in Digital Topology}
\author{Laurence Boxer
\thanks{
    Department of Computer and Information Sciences,
    Niagara University,
    Niagara University, NY 14109, USA;
    and Department of Computer Science and Engineering,
    State University of New York at Buffalo.
    email: boxer@niagara.edu
}
}

\date{ }
\maketitle{}

\begin{abstract}
Two objects may be close in the Hausdorff metric, yet have very different 
geometric and topological properties. We examine other methods of
comparing digital images such that objects close in each of these measures 
have some similar geometric or topological property. Such measures may be combined
with the Hausdorff metric to yield a metric in which close images are similar with
respect to multiple properties.

Key words and phrases: digital topology, digital image, Hausdorff metric

MSC: 54B20
\end{abstract}


\section{Introduction}
A key question in digital image processing is whether two digital 
images $A$ and $B$ represent the same object. If, after magnification or shrinking and 
translation, copies $A'$ and $B'$ of the respective images
have been scaled to approximately the same size and are located in approximately the same
position, a Hausdorff metric $H$ may be employed: if $H(A',B')$ is small, then perhaps
$A$ and $B$ represent the same object; if $H(A',B')$ is large, then 
$A$ and $B$ probably do not represent the same object. However, the Hausdorff metric
is very crude as a measure of similarity. In this paper, we consider other
comparisons of digital images.

\section{Preliminaries}
Much of this section is quoted or paraphrased from~\cite{BxSt16}.

We use $\N$ to indicate the set of natural numbers, $\Z$ for the set of integers, and
$\R$ for the set of real numbers.

\subsection{Adjacencies}
A digital image is a graph $(X,\kappa)$, where $X$ is a subset of $\Z^n$ for
some positive integer~$n$, and $\kappa$ is an adjacency relation for the points
of~$X$. The $c_u$-adjacencies are commonly used.
Let $x,y \in \Z^n$, $x \neq y$, where we consider these points as $n$-tuples of integers:
\[ x=(x_1,\ldots, x_n),~~~y=(y_1,\ldots,y_n).
\]
Let $u \in \Z$,
$1 \leq u \leq n$. We say $x$ and $y$ are 
{\em $c_u$-adjacent} if
\begin{itemize}
\item There are at most $u$ indices $i$ for which 
      $|x_i - y_i| = 1$.
\item For all indices $j$ such that $|x_j - y_j| \neq 1$ we
      have $x_j=y_j$.
\end{itemize}
Often, a $c_u$-adjacency is denoted by the number of points
adjacent to a given point in $\Z^n$ using this adjacency.
E.g.,
\begin{itemize}
\item In $\Z^1$, $c_1$-adjacency is 2-adjacency.
\item In $\Z^2$, $c_1$-adjacency is 4-adjacency and
      $c_2$-adjacency is 8-adjacency.
\item In $\Z^3$, $c_1$-adjacency is 6-adjacency,
      $c_2$-adjacency is 18-adjacency, and $c_3$-adjacency
      is 26-adjacency.
\end{itemize}

We write $x \adj_{\kappa} x'$, or $x \adj x'$ when $\kappa$ is understood, to indicate
that $x$ and $x'$ are $\kappa$-adjacent. Similarly, we
write $x \adjeq_{\kappa} x'$, or $x \adjeq x'$ when $\kappa$ is understood, to indicate
that $x$ and $x'$ are $\kappa$-adjacent or equal.

A subset $Y$ of a digital image $(X,\kappa)$ is
{\em $\kappa$-connected}~\cite{Rosenfeld},
or {\em connected} when $\kappa$
is understood, if for every pair of points $a,b \in Y$ there
exists a sequence $\{y_i\}_{i=0}^m \subset Y$ such that
$a=y_0$, $b=y_m$, and $y_i \adj_{\kappa} y_{i+1}$ for $0 \leq i < m$.

\subsection{Digitally continuous functions}
The following generalizes a definition of~\cite{Rosenfeld}.

\begin{definition}
\label{continuous}
{\rm ~\cite{Boxer99}}
Let $(X,\kappa)$ and $(Y,\lambda)$ be digital images. A single-valued function
$f: X \rightarrow Y$ is $(\kappa,\lambda)$-continuous if for
every $\kappa$-connected $A \subset X$ we have that
$f(A)$ is a $\lambda$-connected subset of $Y$. $\Box$
\end{definition}

When the adjacency relations are understood, we will simply say that $f$ is \emph{continuous}. Continuity can be expressed in terms of adjacency of points:
\begin{thm}
{\rm ~\cite{Rosenfeld,Boxer99}}
A function $f:X\to Y$ is continuous if and only if $x \adj x'$ in $X$ implies 
$f(x) \adjeq f(x')$. \qed
\end{thm}

See also~\cite{Chen94,Chen04}, where similar notions are referred to as {\em immersions}, {\em gradually varied operators},
and {\em gradually varied mappings}.

\subsection{Pseudometrics and metrics}
\begin{definition}
{\rm \cite{Dugundji}}
Let $X$ be a nonempty set. Let $d: X^2 \to [0,\infty)$ be a function such that
for all $x,y,z \in X$,
\begin{itemize}
    \item $d(x,y) \ge 0$;
    \item $d(x,x) = 0$;
    \item $d(x,y) = d(y,x)$; and
    \item $d(x,z) \le d(x,y) + d(y,z)$.
\end{itemize}
Then $d$ is a {\em pseudometric} for $X$. If, further, $d(x,y) = 0$ implies $x=y$ then
$d$ is a metric for $X$.
\end{definition}

Pseudometrics that can be applied to pairs $(A,B)$ of
nonempty subsets of a digital image $X$ include
the absolute values of the differences in their
\begin{itemize}
    \item deviations from convexity. Several such deviations are discussed 
          in~\cite{Stern,Boxer93}, for each of which it was shown that two objects can be
          ``close" in the Hausdorff metric yet quite different with respect to the
          deviation from convexity. These can be adapted to digital images with respect
          to digital convexity as defined in~\cite{BxConvex}.
     \item Euler characteristics. I.e., the function 
          \[ s_{\chi}(A,B) = |\chi(A) - \chi(B)|,
          \]
          where $\chi(X)$ is the Euler characteristic of $(X,\kappa)$, 
          is a pseudometric for digital images in $\Z^n$.
          An improper definition of the Euler characteristic
          for digital images was given in~\cite{Han07}. An appropriate definition is
          given in~\cite{BKO11}. 
    \item Lusternik-Schnirelman category $cat_{\kappa}(X)$~\cite{BV}. I.e., the function
          \[ s_{LS,\kappa}(A,B) = |cat_{\kappa}(A) - cat_{\kappa}(B)|,
          \]
          where $cat_{\kappa}(X)$ is the Lusternik-Schnirelman category of $(X,\kappa)$,
          is a pseudometric for digital images in $\Z^n$.
    \item diameters. This is discussed below.
\end{itemize}

The following is easily verified and extends an assertion of~\cite{Boxer93}.

\begin{lem}
\label{addPseudos}
Let $\Delta_i: X^2 \to [0,\infty)$ be a pseudometric,
$1 \le i \le n$. Then
$D = \sum_{i=1}^n \Delta_i: X^2 \to [0,\infty)$ is a 
pseudometric. Further, if at least one of the $\Delta_i$ is 
a metric, then $D$ is a metric.
\end{lem}

Here we mention metrics we use in this paper for $\R^n$ or $\Z^n$. Let
$x = (x_1, \ldots, x_n)$, $y=(y_1, \ldots, y_n)$.
\begin{itemize}
    \item Let $p \ge 1$. The $\ell_p$ {\em metric} for $\R^n$ is given by
    \[ d_p(x,y) = \left( \sum_{i=1}^n |x_i - y_i|^p \right )^{1/p}.
    \]
    The special case $p=1$ gives the {\em Manhattan} or {\em city block} metric $d_1: (\R^n)^2 \to [0,\infty)$, given by
          \[ d_1(x,y) = \sum_{i=1}^n |x_i - y_i|.
          \]
    The special case $p=2$ gives the {\em Euclidean metric} $d_2: (\R^n)^2 \to [0,\infty)$, given by
          \[ d_2(x,y) = \left( \sum_{i=1}^n |x_i - y_i|^2 \right )^{1/2}.
          \]
    \item The {\em shortest path metric}~{\rm \cite{Han05}}:
          Let $(X,\kappa)$ be a connected digital image. For $x,y \in X$, let
\[ d_{\kappa}(x,y) = \min\{n \, | \mbox{ there is a $\kappa$-path of length $n$ in $X$ from $x$ to $y$}.\}
\]
    \item The {\em Hausdorff metric based on a metric} $d$~\cite{Nadler}:
          Let $d: X^2 \to [0,\infty)$ be a metric where 
          $X \subset \R^n$. The Hausdorff metric for
          nonempty bounded and closed subsets
          $A$ and $B$ of $X$ (hence, in the case 
          $X \subset \Z^n$, finite subsets of $X$) 
          based on $d$ is
          \[ H(A,B) = \min
          \left \{ \begin{array}{l}
          \varepsilon > 0 ~ | ~ \forall (a,b) \in A \times B,~  \exists~ (a',b') \in A \times B \\
                \mbox{ such that }  \varepsilon \ge d(a,b') \mbox{ and }  \varepsilon \ge d(a',b)
          \end{array} \right \}.
          \]
\end{itemize}

We make the following modification of the Hausdorff metric based on $d_{\kappa}$
as presented in~\cite{Vergili}.
\begin{definition}
Let $X \subset \Z^n$, $\emptyset \neq A \subset X$, $\emptyset \neq B \subset X$. Let
$\kappa$ be an adjacency on $X$. Then
\[ H_{(X,\kappa)}(A,B) = \min 
   \left \{ \begin{array}{l}
   \varepsilon \ge 0 ~ | ~ \forall (a,b) \in A \times B,~
                       \exists (a',b') \in A \times B  \\ \mbox{ such that }
 \mbox{there are $\kappa$-paths in $X$ of length} \le \varepsilon 
  \\ \mbox{ from $a$ to $b'$ and from $b$ to $a'$}
   \end{array} \right \}.
\]
\end{definition}
In the version of the Hausdorff metric based on $d_{\kappa}$
in~\cite{Vergili}, $X=\Z^n$. We show below that we can 
get very different results for the more general situation
$\emptyset \neq X \subset \Z^n$.

We use the notations $H_d$ for the Hausdorff metric 
based on the metric $d$, $H_p$ for the Hausdorff metric 
based on the $\ell_p$ metric $d_p$ (i.e., $H_p = H_{d_p}$),
and $H_{(X,\kappa)}$ for the Hausdorff metric based on
$d_{\kappa}$ for subsets of $X$ (i.e.,
$H_{\kappa} = H_{d_{\kappa}}$).

Another metric from classical topology that is easily adapted to digital topology is
Borsuk's {\em metric of continuity}~\cite{Borsuk54, Borsuk67} based on a metric $d$ 
which is typically, but not necessarily,
the Euclidean metric. For digital images $(X,\kappa)$ and $(Y,\kappa)$ in $\Z^n$,
define the metric of continuity $\delta_d(X,Y)$ as the greatest lower bound of numbers
$t>0$ such that there are $\kappa$-continuous $f: X \to Y$ and $g: Y \to X$ with
\[ d(x, f(x)) \le t \mbox{ for all } x \in X \mbox{ and } d(y, g(y)) \le t 
    \mbox{ for all } y \in Y.
\]

\begin{prop}
\label{HausAndContinuityMetrics}
Given finite digital images $(X,\kappa)$ and $(Y,\kappa)$ in $\Z^n$ and a metric $d$ for $\Z^n$,
$H_d(X,Y) \le \delta_d(X,Y)$.
\end{prop}

\begin{proof}
This is largely the argument of the analogous assertion
in~\cite{Borsuk54}. Let $u=H_d(X,Y)$. Since $X$ and $Y$ are finite,
without loss of generality, there
exists $x_0 \in X$ such that $u = \min\{y \in Y~|~d(x_0,y)\}$. Then for
all $\kappa$-continuous $f: X \to Y$, $d(x_0,f(x_0)) \ge u$.
Therefore, $\delta_d(X,Y) \ge u$.
\end{proof}

An example for which the inequality of Proposition~\ref{HausAndContinuityMetrics} is
strict is given in the following.

\begin{thm}
Let $X = \{(x,y) \in \Z^2 ~|~ |x| = n \mbox{ or } |y| = n\}$.
Let $Y = X \setminus \{(n,n)\}$. Then, using the Manhattan metric for $d$ and 
$\kappa = c_1$, we have
$H_1(X,Y) = 1$ and $\delta_d(X,Y) \ge 2n-1$.
\end{thm}

\begin{proof}
It is clear that $H_1(X,Y) = 1$.

Notice there is an isomorphism $F: (Y, c_1)$ to a subset
of $(\Z,c_1)$. Let
$f: X \to Y$ be $c_1$-continuous. By~\cite{Boxer10}, there is a pair of antipodal points
$P, -P \in X$ such that $|F \circ f(P) - F \circ f(-P)| \le 1$. Since $F$ is an
isomorphism, we must have $f(P) \adjeq_{c_1} f(-P)$. We will show that
either $d(P,f(P)) \ge 2n-1$ or $d(-P,f(-P)) \ge 2n-1$, as follows.

If $P=(n,u)$ then $-P=(-n,-u)$. Then:
\begin{itemize}
    \item If $f(P) = (n,-n)$ then $f(-P) \in \{(n-1,-n), (n,-n), (n,-n+1)\}$, so 
          $d(-P,f(-P)) \ge 2n-1$.
    \item Note $(n,n) \not \in Y$ so $f(P)$ cannot equal
          $(n,n)$.
    \item If $f(P) = (n,v)$ for $|v| < n$ then $f(-P) \in \{(n,v-1), (n,v), (n,v+1)\}$,
          so $d(-P,f(-P)) \ge 2n$.
\end{itemize}
The cases $P=(-n,u)$, $P=(w,n)$, and $P=(w,-n)$ are similar.
Thus $\delta_d(X,Y) \ge 2n-1$.
\end{proof}

We say the {\em diameter} of a nonempty bounded set $A \subset \R^n$ with respect to a metric $d$ is
\[ diam_d(A) = \max\{d(a,b) \, | \, a,b \in A\}.
\]
We will use the notations $diam_p$ for $diam_{d_p}$, and 
$diam_{\kappa}$ for $diam_{d_{\kappa}}$.

We define a function $s_d$ for pairs of nonempty bounded sets in $\R^n$ by
\[ s_d(A,B) = |diam_d(A) - diam_d(B)|.
\]
We use notations $s_p$ for $s_{d_p}$, and $s_{\kappa}$ for $s_{d_{\kappa}}$.

The following is easily verified.
\begin{lem}
The function $s_d$ is a pseudometric.
\end{lem}

\section{Comparing (pseudo)metrics on digital images}
In this section, we compare the use of some of the (pseudo)metrics discussed above.

\begin{thm}
\label{closeDiam}
Let $A$ and $B$ be nonempty, bounded subsets of $\R^n$. Let $H_p$ be the Hausdorff metric
based on the $\ell_p$ metric $d_p$ and suppose $H_p(A,B) \le m$. Then
$s_p(A,B) \le 2m$.
\end{thm}

\begin{proof}
There exist $a,a' \in A$ such that $d_p(a,a') = diam_p(A)$. There exist
$b, b' \in B$ such that $d_p(a,b) \le m$ and $d_p(a',b') \le m$. So
\[ diam_p(A) = d_p(a,a') \le d_p(a,b) + d_p(b,b') + d_p(b',a') \le 
                     m + diam_p(B) + m
                     \]
\[ = diam_p(B) + 2m.
\]
Similarly, $diam_p(B) \le diam_p(A) + 2m$. The assertion follows.
\end{proof}

\begin{figure}
    \centering
    \includegraphics[height=2in]{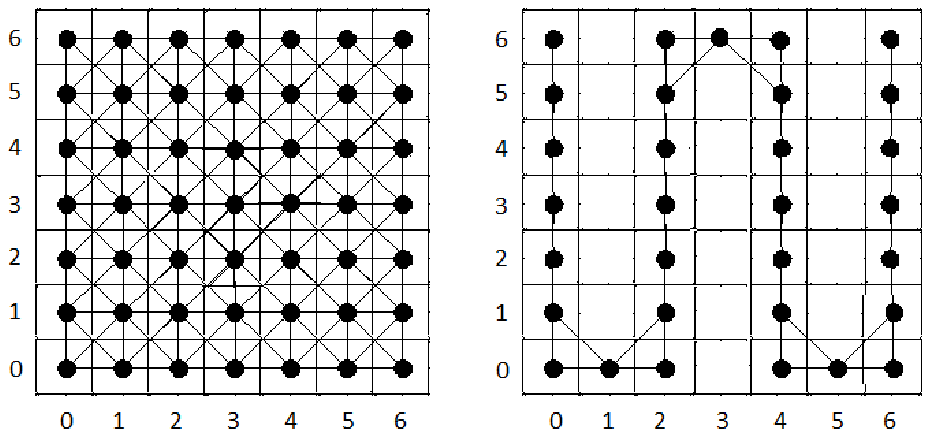}
    \caption{Left: $Q = [0,n]_{\Z}^2$ (here, $n=6$). \newline
             Right: \newline
             $S = Q \setminus 
             \left [ (\bigcup_{k \in \Z} \{4k+1\} \times [1,n]_{\Z} ) \cup 
             (\bigcup_{k \in \Z} \{4k+3\} \times [0,n-1]_{\Z} )
             \right ]$ \newline
             (here, $n=6$). \newline
            $Q$ and $S$ are within 1 with respect to the Hausdorff metric based on the Manhattan
            metric; however, they differ considerably with respect to diameter in the 
            shortest path metric.
          }
    \label{fig:squareSnake}
\end{figure}

By contrast, we have the following.

\begin{exl}
Let $n \in \N$ such that $n$ is even. 
Let $Q = [0,n]_{\Z}^2$. Let 
 \[ S = Q \setminus 
             \left [ (\bigcup_{k \in \Z} \{4k+1\} \times [1,n]_{\Z} ) \cup 
             (\bigcup_{k \in \Z} \{4k+3\} \times [0,n-1]_{\Z} )
             \right ].
             \]
(See Figure~\ref{fig:squareSnake}.) Then $s_1(Q,S)=0$, but
while $diam_{c_1}(Q) = 2n$, we have $diam_{c_1}(S)= n + n(1+n/2)$.
Thus $s_{c_1}(Q,S) = n^2 / 2$.
\end{exl}

\begin{proof}
It is easy to see that both $Q$ and 
$S$ have diagonally opposed points that are maximally 
distant in the $d_1$ metric. Therefore, 
$diam_1(S) = diam_1(Q) = 2n$, so $s_1(Q,S) = 0$.

Diagonally opposed points of $Q$ are maximally separated with respect to $d_{c_1}$,
so $diam_{c_1}(Q) = 2n$. Maximally separated points of $S$ with respect to $d_{c_1}$ are
\[ \begin{array}{ll}
    (0,n) \mbox{ and } (n,n) & \mbox{if } n=4k+2 \mbox{ for some } k \in \Z; \\
    (0,n) \mbox{ and } (n,0) & \mbox{if } n=4k \mbox{ for some } k \in \Z.
\end{array}
\]
In either case, the unique shortest $c_1$-path between maximally separated points
requires $n$ horizontal steps. The number of vertical steps is computed as follows. There are
$1+n/2$ vertical line segments that must be traversed, each of length $n$, so the number of
vertical steps is $n(1+n/2)$. Thus the number of steps between maximally separated members of $S$ is
$diam_{c_1}(S)= n + n(1+n/2)$.

Hence for $\kappa = c_1$ we have 
$s_{\kappa}(Q,S) = |n + n(1+n/2) - 2n| =  n^2/2$.
\end{proof}

We do not get an analog of Theorem~\ref{closeDiam} by using 
the Hausdorff metric based on an adjacency $\kappa$
instead of $H_p$. This is shown in the following example.

\begin{exl}
\label{rectAndC}
Let $A = [0,n]_{\Z} \times [0,2]_{\Z}$. Let $B = A \setminus ([1,n]_{\Z} \times \{1\})$. (See
Figure~\ref{fig:6x2-and-C}.) Then
$H_1(A,B) = 1$. However, we have the following.
\begin{itemize}
    \item For $\kappa = c_1$, $diam_{\kappa}(A) = n+2$ and $diam_{\kappa}(B) = 2n + 2$,
          so $s_{\kappa}(A,B) = n$.
    \item For $\kappa = c_2$, $diam_{\kappa}(A) = n$ and $diam_{\kappa}(B) = 2n$, so     
          $s_{\kappa}(A,B) = n$.
\end{itemize}
\end{exl}

\begin{figure}
    \centering
    \includegraphics[height=1.5in]{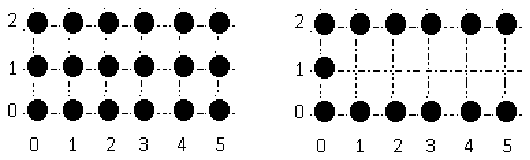}
    \caption{Digital images $A$ (left) and $B$ (right) for Example~\ref{rectAndC},
             using $n=5$. Using the shortest path metric and either
             $\kappa = c_1$ or $\kappa = c_2$,
             maximally distant points in $A$ are $(0,0)$ and $(n,2)$,
             and  maximally distant points in $B$ are $(n,0)$ and $(n,2)$.
           }
    \label{fig:6x2-and-C}
\end{figure}

\begin{thm}
\label{c2AndEuclidean}
Let $A$ and $B$ be  finite, nonempty $c_u$-connected 
subsets of a $c_u$-connected subset $X$ of $\Z^n$, 
where $1 \le u \le n$.
Suppose we have
$H_{(X,c_u)}(A,B) \le m$ for some $m \in \N$. 
Then $H_p(A,B) \le m u^{1/p}$.
\end{thm}

\begin{proof}
By hypothesis, given $x \in A$ and $y \in B$, there exist $x' \in A$, $y' \in B$,
and $c_u$-paths $P$ from $x$ to $y'$ and $Q$ from $y$ to $x'$ in $X$ such that
each of $P$ and $Q$ has length of at most $m$. Since each $c_u$-adjacency corresponds to
a Euclidean distance of at most $u^{1/p}$, it follows that
$d_p(x,y') \le mu^{1/p}$ and $d_p(y,x') \le mu^{1/p}$. It follows that
$H_p(X,Y) \le mu^{1/p}$.
\end{proof}

We do not get a converse for Theorem~\ref{c2AndEuclidean}, as the following shows.
\begin{exl}
\label{H-for-l_p_and-shortestPath}
Let $B = [0,n]_{\Z} \times [0,2]_{\Z} \setminus ([1,n]_{\Z} \times \{1\})$ as in
Example~\ref{rectAndC}. (See Figure~\ref{fig:6x2-and-C}.) Let
$C = [0,n]_{\Z} \times \{0\} \subset B$. Then
$H_1(B,C) = H_2(B,C) = 2$. However,
$H_{(B,c_1)}(B,C) = n+2$ and $H_{(B,c_2)}(B,C) = n+1$.
\end{exl}

\begin{proof}
It is easy to see that $H_1(B,C) = H_2(B,C) = 2$.

Since $C \subset B$, finding a Hausdorff distance between $B$ and $C$ comes down to
considering a furthest point of $B$ from $C$.
With respect to $\kappa=c_1$ and also with respect to $\kappa=c_2$, the furthest
point of $B$ from $C$ in the shortest path metric is $b=(n,2)$ and its closest point of $C$ 
is $c=(0,0)$. Since $d_{c_1}(b,c)=n+2$ and $d_{c_2}(b,c)=n+1$, the assertion follows.
\end{proof}

Roughly, it appears that the great differences found in Examples~\ref{rectAndC}
and~\ref{H-for-l_p_and-shortestPath}, between measures based in $\ell_p$ metrics and
measures based on the shortest path metric, are due to significant deviations from
convexity.
If we consider $H_{(X,c_i)}(A,B)$ for a set $X$ such as a digital cube, we
may find $H_p$ and $H_{(X,c_p)}$ are more alike, as we see below.

\begin{prop}
\label{H_1AndH_shortestPath}
Let $A \neq \emptyset \neq B$, $A \cup B \subset J=[0,n]_{\Z}^2$. Then
$H_1(A,B) = H_{(J, c_1)}(A,B)$.
\end{prop}

\begin{proof}
Let $n = H_1(A,B)$. Let $x \in A$. Then there exists $y \in B$ such that
$d_1(x,y) \le n$. By definition of $d_1$, it follows that there is a $c_1$-path
in $J$ of length at most $n$ from $x$ to $y$. Similarly, given $u$ in $B$, there is a $c_1$-path
in $J$ of length at most $n$ from $u$ to a point $v \in A$. Therefore,
$H_{(J, c_1)}(A,B) \le n = H_1(A,B)$.

Now let $n = H_{(J, c_1)}(A,B)$. Then given $x \in A$, there is a $c_1$-path in $J$ of
length at most $n$ from $x$ to some $y \in B$. Similarly, 
given $u \in B$, there is a $c_1$-path in $J$ of
length at most $n$ from $u$ to some $v \in A$. Since every $c_1$ adjacency represents
a $d_1$ distance of 1, it follows that $d_1(x,y) \le n$ and $d_1(u,v) \le n$.
Thus $H_1(A,B) \le n = H_{(J, c_1)}(A,B)$. The assertion follows.
\end{proof}

Using the observation that a $c_u$-adjacency in $\Z^r$, $1 \le u \le r$,
represents a $d_p$ distance between the adjacent points 
that is between $1$ and $u^{1/p}$,
we can generalize the argument used to prove Proposition~\ref{H_1AndH_shortestPath},
getting the following.

\begin{thm}
\label{H_kAndH_shortestPath}
Let $A \neq \emptyset \neq B$, $A \cup B \subset J=[0,n]_{\Z}^v$. Then for $1 \le u \le v$,
$H_{(J,c_u)}(A,B) \le u^{1/p} \cdot H_{(J, c_u)}(A,B)$.
\end{thm}

\section{Further remarks}
The Hausdorff metric is often used to compare objects $A$ and $B$. It is
easy to compute efficiently~\cite{Shon89, BxMi} 
and gives a good indication of how well each of its arguments approximates the other
with respect to location.

However, two objects may be close in the Hausdorff metric and yet have very different
geometric or topological properties. Lemma~\ref{addPseudos} tells us that
by adding other pseudometrics or metrics, such as those we have discussed, to the
Hausdorff metric, we can get another metric in which closeness is more likely to
validate the parameters as digital images representing the same physical object.

The suggestions and corrections of an anonymous 
reviewer are gratefully acknowledged.

\end{document}